\documentclass[12pt]{article}
\usepackage{amsmath, amssymb,amsthm,a4wide}
\usepackage{authblk}       

\usepackage[colorlinks=true]{hyperref}
\newtheorem{lemma}{Lemma}
\newtheorem{remark}{Remark}
\newtheorem{corollary}{Corollary}

\newcommand{\cL}{\mathcal{L}}
\newcommand{\cK}{\mathcal{K}}

\newcommand{\cI}{\mathcal{I}}
\newcommand{\cH}{\mathcal{H}}
\newcommand{\coR}{\overline{\mathcal{R}}}
\newcommand{\coK}{\overline{\mathcal{K}}}
\newcommand{\coM}{\overline{\mathcal{M}}}
\newcommand{\orho}{\overline{\rho}}
\newcommand{\otau}{\overline{\tau}}
\newcommand{\opsi}{\overline{\psi}}

\newcommand{\ba}{\boldsymbol{a}}
\newcommand{\bb}{\boldsymbol{b}}

\newcommand{\bA}{\boldsymbol{A}}
\newcommand{\bB}{\boldsymbol{B}}

\newcommand{\bF}{\boldsymbol{F}}
\newcommand{\bG}{\boldsymbol{G}}
\newcommand{\bH}{\boldsymbol{H}}
\newcommand{\bI}{\boldsymbol{I}}

\newcommand{\bL}{\boldsymbol{L}}
\newcommand{\bM}{\boldsymbol{M}}
\newcommand{\bN}{{\boldsymbol{N}}}

\newcommand{\bW}{\boldsymbol{W}}
\newcommand{\bX}{\boldsymbol{X}}
\newcommand{\bY}{\boldsymbol{Y}}

\newcommand{\bSx}{\boldsymbol{\sigma_{\! x}}}
\newcommand{\bSy}{\boldsymbol{\sigma_{\! y}}}
\newcommand{\bSz}{\boldsymbol{\sigma_{\! z}}}
\newcommand{\bSp}{\boldsymbol{\sigma_{\!\text{\bf  +}}}}
\newcommand{\bSm}{\boldsymbol{\sigma_{\!\text{\bf -}}}}

\newcommand{\bSzA}{\boldsymbol{\sigma_{\! z}^{A}}}
\newcommand{\bSpA}{\boldsymbol{\sigma_{\!\text{\bf  +}}^{A}}}
\newcommand{\bSmA}{\boldsymbol{\sigma_{\!\text{\bf -}}^{A}}}
\newcommand{\bSzB}{\boldsymbol{\sigma_{\! z}^{B}}}

\newcommand{\Hint}{\boldsymbol{H}_{\!\rm{int}}}
\newcommand{\Lint}{\mathcal{L}_{\!\rm{int}}}
\newcommand{\cD}{\mathcal{D}}
\newcommand{\oa}{\overline{\ba}}

\newcommand{\ket}[1]{\left|#1\right\rangle}
\newcommand{\bra}[1]{\left\langle #1\right|}
\newcommand{\bket}[1]{\left\langle #1 \right\rangle}
\newcommand{\braket}[2]{\left\langle #1\right|\left.#2 \right\rangle}
\newcommand{\ketbra}[2]{\left|#1 \right\rangle\!\left\langle #2 \right|}
\newcommand{\dotex}{\frac{d}{dt}}

\newcommand{\rs}{\rho_s}
\newcommand{\tr}{{\rm{tr}}}

\newcommand{\ra}[1]{{ #1}}

\begin{document}

\title{Towards generic adiabatic elimination for bipartite open  quantum systems}

\author[*]{R. Azouit}
\author[$\dag$]{F. Chittaro}
\author[$\ddag$]{A. Sarlette}
\author[*]{P. Rouchon}

\affil[*]{\small Centre Automatique et  Syst\`{e}mes, Mines-ParisTech, PSL Research University.
60 Bd Saint-Michel, 75006 Paris, France.}
\affil[$\dag$]{INRIA Paris, 2 rue Simone Iff, 75012 Paris, France; and Ghent University, Belgium.}
\affil[$\ddag$]{Aix Marseille Universit\'{e}, CNRS, ENSAM, LSIS UMR 7296, 13397 Marseille, France, and Universit\'e de Toulon,
CNRS, LSIS UMR 7296, 83957 La Garde, France}
\date{}

\maketitle
\begin{abstract}
We consider a composite open quantum system consisting of a fast subsystem  coupled to a slow one. Using the time-scale separation, we develop an adiabatic elimination technique to derive at any order   the reduced model describing the slow subsystem. The method, based on an asymptotic expansion and geometric singular perturbation theory, ensures the physical interpretation of the reduced  second-order model by giving the reduced dynamics in a Lindblad form and the state reduction in Kraus map form. We give explicit  second-order formulas for Hamiltonian or cascade coupling between the two subsystems. These formulas can be used to engineer, via a careful choice of the fast subsystem, the  Hamiltonian and Lindbald operators  governing the  dissipative dynamics  of the slow subsystem.
\end{abstract}

\section{Introduction}
As any quantum system interacts with an environment, the rigorous way to describe its evolution would be a Schr\"odinger equation including the environment. Studying this whole system is usually not possible due to the complexity of the environment. Therefore, using typical assumptions on the environment one can get rid of it by some Markov approximations and obtain
a Lindblad master equation \cite{BreuerPetruccioneBook} governing only the system of interest. A similar situation arises naturally within a quantum system composed of several interacting subsystems: one may want to get the dynamics of a particular subsystem of interest by using some assumptions to get rid of the other ones. When the  subsystem of interest has much slower dissipation rates (i.e.~time scales) than the other ones and is weakly  coupled to the other ones, such procedure is known as adiabatic elimination.

This kind of setting is related to \ra{reservoir engineering \cite{PoyatCZ1996PRL, CarvalhoPRL01, LeghtasSC15}} where the goal is to design the dynamics of a target subsystem by properly choosing its interaction with the other subsystems. The interaction is engineered in such a way that the dynamics of the  target subsystem is, after adiabatic elimination of the other subsystems, the desired one. Rigorously computing the reduced dynamics associated with the subsystem of interest is not straightforward. We here propose a systematic adiabatic elimination procedure that preserves the quantum structure.  As illustrated by the following example, the resulting reduced dynamics may include some non-intuitive phenomena.

Consider a pair of qubits. The first qubit, labelled by $A$, is subject to strong energy loss (life time of the excited state $1/\gamma$ ) and is driven by a coherent resonant field of amplitude $u \geq 0$. The second qubit, labelled by $B$, is isolated from the environment but is weakly coupled to the first one by a dispersive Hamiltonian of strength $\chi$. This setting is an abstraction of a target qubit $B$ coupled to a spurious ``TLS'', i.e.~a strongly decoherent two-level system with $\gamma \gg |\chi|$, see e.g.~\cite{lisenfeld2015observation}.
The density operator $\rho$ of the joint system is governed by the following master equation:
\begin{align*}
\frac{d \rho}{dt} = u[\bSpA - \bSmA\;,\;\rho] + \gamma (\bSmA \rho \bSpA - \tfrac{\bSpA \bSmA \rho + \rho \bSpA \bSmA}{2})
 - i \chi [\bSzA \otimes \bSzB, \rho] \;
\end{align*}
where $\bSmA$ is the lowering operator for qubit $A$, $\bSpA=(\bSmA)^\dag$, and $\bSzA$, $\bSzB$ are the usual Pauli operators  for qubits $A$ and $B$ respectively. Since $\gamma \gg |\chi|$, the effect of the ``TLS'' qubit $A$ on the target qubit $B$ can be obtained by adiabatically eliminating qubit $A$. A direct application of the formula~\eqref{eq:L1L2dispersive} that we obtain below, shows that up to third  order terms in $\epsilon = |\chi|/\gamma$,  the reduced dynamics of qubit $B$ are:
\begin{align}\label{eq:example1}
\frac{d\rho}{dt} = i\chi \frac{\gamma^2}{\gamma^2+8u^2} [\bSzB,\rho] + 64  \gamma \chi^2  u^2 \, \frac{(\gamma^2 + 2u^2)}{(\gamma^2+8u^2)^3} \;
\big(\bSzB \rho \bSzB - \rho\big) \, .
\end{align}
As one may intuitively expect, qubit $A$ induces on qubit $B$ a deterministic phase shift, corresponding to the mean value of $\bSzA$, and some phase loss related to the uncertainty in $\bSzA$. When $u$ is small, as expected, qubit $B$ suffers little phase loss as qubit $A$ is close to the ground state, which is an eigenstate of $\bSzA$. However, more surprisingly, in the presence of strong drive (i.e large $u$) the phase loss also vanishes. This is due to the fact that the system is then subject to a sort of dynamical decoupling \cite{ViolaPRL99}.

The present paper thus provides a systematic method for deriving the Lindblad form of reduced dynamics like \eqref{eq:example1}, up to third order terms, for general bipartite systems with two times-scales. For a Hamiltonian coupling between the two subsystems, this second order Lindblad form is given by the generic formula~\eqref{eq:L1L2Hamiltonnian}. For a cascade coupling, it is given by~\eqref{eq:L1L2Cascade}.

Throughout the literature, adiabatic elimination formulas for open quantum systems have been obtained for several particular examples separately. A specific atom-optics system is analyzed in \cite{AtkinsPRA03}; for lambda systems see e.g. \cite{BrionJPA07}, \cite{mirrahimi-rouchonieee09}, \cite{ReiteS2012PRA}; systems with Gaussian dynamics are investigated in \cite{CernotikPRA15}. In contrast, generic techniques for adiabatic elimination have attracted much less attention. Treating the Lindblad master as a standard linear system, Kessler has developed in  \cite{Kessl2012PRA} a generalization of the Schrieffer-Wolff formalism. \ra{In \cite{GoughJSP07, BoutenJFA08}, the authors derive a reduced dynamics for systems described by Quantum Stochastic Differential Equations in the limit where the speed of the fast system goes to infinity. Our aim is to go further and precisely characterise the order of the approximation. This is motivated by the fact that in some recent quantum experiments aiming via reservoir engineering at strong indirect stabilization of quantum systems \cite{LeghtasSC15}, a better knowledge of the order of validity of such approximations is becoming necessary : the increase in the accuracy of experiments implies the need to compute higher order model in order to properly describe the slow dynamics.} Our aim with the present paper is to provide formulas with improved direct applicability, control on the order of approximation, and explicit expression of how the reduced system is embedded in the full Hilbert space.

We focus on bipartite open quantum systems. Using a geometric approach, we treat the slow dynamics as a perturbation of the fast one. This leads to an asymptotic expansion in powers of the small parameter describing the time-scales separation and readily allows one to select the order of the approximation. Moreover, the specificity of our approach compared to standard abstract systems theory, is to preserve the structural properties that ensure the physical interpretation of the reduced model: first, the reduced density operator is governed by a Lindblad master equation; and second, the mapping from the reduced density operator to the density operator of the entire system is a Kraus map \cite{ChoiLAA1975} i.e.~completely positive trace preserving. These two features ensure the preservation of a density operator interpretation. The computation of the Kraus map allows to estimate for instance the residual entanglement, inside the protected subspace, between the two subsystems of a reservoir engineering setup, which may have important practical consequences. To our knowledge, such care about the mapping between the complete and reduced systems is new in adiabatic elimination techniques.

Our approach builds on the framework of center manifold theory \cite{carr-book} and geometric singular perturbations theory \cite{Fenichel79} to obtain recurrence relations between the approximations at different orders. At the first order, we readily retrieve the so-called Zeno Hamiltonian for any type of coupling between the fast and slow subsystems. Then, focusing on two standard forms of coupling namely Hamiltonian interaction and cascade interaction, we derive formulas for the second order approximations while ensuring a Lindblad form for the reduced dynamics and a mapping in Kraus map form. The formulas show that the Hamiltonian interaction with a decohering subsystem leads to decoherence at the second order. \ra{Before performing any detailed calculations, we obtain} the general structural result that the number of decoherence channels is equal to the minimal number of tensor-product terms required to express the interaction Hamiltonian.

The paper is organized as follows. In section \ref{sec:BipartSyst} we present the considered system types, our general adiabatic elimination approach, and the results on first order approximation. Section \ref{sec:HamSyst} is devoted to the computation of the second order approximation in a Lindblad form for a Hamiltonian coupling between the fast and slow subsystems. The typical cases of dispersive or resonant interaction are mentioned explicitly (for readers interested mainly in the final formulas: see equations \eqref{eq:L1L2dispersive} and \eqref{eq:L1L2resonnant}). In Section \ref{ssec:Ex2photon}, the use of these formulas is illustrated by obtaining new explicit expressions for the residual imperfections on a reservoir engineering setup with two-photon pumping \cite{theseJoachimCohen,TouzardOngoing}. Section \ref{sec:CascSyst} investigates a different kind of coupling, namely the cascade interaction. We derive the second order approximation in a Lindblad form provided some technical condition holds.

\paragraph{Acknowledgements:} The authors want to thank Zaki Leghtas and Mazyar
Mirrahimi for useful discussions about this subject.

\paragraph{Notations:}
The underlying Hilbert spaces are denoted by $\cH_{\bullet}$ with different subscripts $\bullet$.
Operators on Hilbert spaces are denoted with bold letters such as $\ba$ or $\bX$.
Super-operators, i.e.~operators acting e.g.~on $\bX$, are denoted via calligraphic capital letters such as $\cL$ or $\cK$.

\section{Adiabatic elimination in bipartite systems} \label{sec:BipartSyst}

Consider on the Hilbert space $\cH_A\otimes \cH_B$ a Lindblad master equation  with a two-time scale  structure:
\begin{equation}\label{eq:CompositeDyn}
\dotex \rho = \cL_A(\rho) + \epsilon  \Lint(\rho) + \epsilon \cL_B(\rho)
\end{equation}
where
\begin{itemize}
  \item $\epsilon$ is a small positive parameter;

  \item $\cL_A$ and $\cL_B$ are Lindbladian  super-operators acting only on $\cH_A$ and $\cH_B$ respectively:
\begin{equation} \label{eq:DefLind_A}
  \cL_{\xi}(\rho)= - i [\bH_\xi,\rho] +\sum_\mu \cD_{\bL_{\xi,\mu}}(\rho), \quad \xi=A,B
\end{equation}
  with  $\bH_\xi$ a Hermitian operator on $\cH_\xi$,  $\bL_{\xi,\mu}$ operators on  $\cH_\xi$ (not necessarily Hermitian)  and with    $\cD_{\bX}$ denoting  the dissipation super-operator associated with operator $\bX$,
 $$
 \cD_{\bX}(\rho)= \bX \rho \bX^\dag - \frac{1}{2} \big(\bX^\dag\bX \rho +  \rho \bX^\dag\bX\big);
 $$
 \item  $\Lint$ is an interaction Lindbladian super-operator  acting on both $\cH_A$ and $\cH_B$.
\end{itemize}

We assume that for any initial  density operator $\rho$ on $\cH_A$, the solution  $\dotex \rho = \cL_A(\rho)$ converges exponentially towards a unique steady-state density operator $\orho_A$.

Consider the solution of~\eqref{eq:CompositeDyn} starting from a density operator $\rho$ on $\cH_A\otimes \cH_B$. For $\epsilon=0$,  this  solution converges towards the separable state $\orho_A\otimes \tr_A(\rho)$, with $\tr_A(\rho)$ the partial trace over $\cH_A$ of the initial density operator on $\cH_A\otimes \cH_B$. For $\epsilon >0$ and small, this solution, after the relaxation time associated with the dynamics on $\cH_A$ governed by $\cL_A$,   remains close  to $\orho_A\otimes \rho_s$ where $\rho_s$ is a slowly evolving density operator on $\cH_B$. We explain here how to express quantitatively, to various orders in $\epsilon$, the fact that the solution of~\eqref{eq:CompositeDyn} remains close to $\orho_A\otimes \rho_s$ and how $\rho_s$ evolves.

 More precisely, we express the solution $\rho$ of~\eqref{eq:CompositeDyn} and the time derivative of $\rho_s$ as  linear functions of $\rho_s$,
by means the linear  super-operators  $\cK$  and $\cL_s$; afterwards, we develop these two super-operators in power of $\epsilon$, thus obtaining an asymptotic expansion of
$\rho$ and the time derivative of $\rho_s$:
\begin{eqnarray}
  \rho  &=& \cK(\rho_s) = \cK_0(\rho_s) + \epsilon \cK_1(\rho_s) + \epsilon^2 \cK_2(\rho_s) + \ldots \label{eq:Kexpansion} \\
  \dotex \rho_s  &=&  \cL_s(\rho_s)=  \cL_0(\rho_s) + \epsilon \cL_1(\rho_s) +  \epsilon^2 \cL_2(\rho_s) + \ldots \label{eq:Lexpansion}
\end{eqnarray}
with
$$
\cK_0(\rho_s)= \orho_A\otimes \rho_s \quad\text{ and }\quad \cL_0(\rho_s)=0.
$$
Plugging these series into~\eqref{eq:CompositeDyn}:
$$
\cL_A(\cK(\rho_s)) + \epsilon  \Lint(\cK(\rho_s)) + \epsilon \cL_B(\cK(\rho_s))=\dotex \rho =  \cK( \cL_s(\rho_s))
$$
gives
\begin{multline*}
\cL_A\left( \epsilon \cK_1(\rho_s) + \epsilon^2 \cK_2(\rho_s) + \ldots \right) + \epsilon  \Lint\left(\cK_0(\rho_s) + \epsilon \cK_1(\rho_s) + \epsilon^2 \cK_2(\rho_s) + \ldots \right) \\+ \epsilon \cL_B\left(\cK_0(\rho_s) + \epsilon \cK_1(\rho_s) + \epsilon^2 \cK_2(\rho_s) + \ldots \right)
\;\;  = \;\;
  \cK_0\left(\epsilon \cL_1(\rho_s) +  \epsilon^2 \cL_2(\rho_s) + \ldots\right) \\+ \epsilon \cK_1\left(\epsilon \cL_1(\rho_s) +  \epsilon^2 \cL_2(\rho_s) + \ldots\right)+ \epsilon^2 \cK_2\left(\epsilon \cL_1(\rho_s) +  \epsilon^2 \cL_2(\rho_s) + \ldots\right) + \ldots
\end{multline*}
where the zero order term vanishes since $\cL_A(\cK_0(\rho_s))\equiv0$.
By identifying terms of same order versus $\epsilon$, it is possible  to compute higher order terms from lower order terms as usual:
\begin{eqnarray}
  \cL_A(\cK_1(\rho_s))+ \Lint(\cK_0(\rho_s))+ \cL_B(\cK_0(\rho_s)) &=& \cK_0(\cL_1(\rho_s)) \label{eq:Order1} \; ,\\
  \cL_A(\cK_2(\rho_s)) +\Lint(\cK_1(\rho_s)) + \cL_B(\cK_1(\rho_s)) &=&  \cK_0(\cL_2(\rho_s))+  \cK_1(\cL_1(\rho_s))
  \; , \label{eq:Order2}
  \\
  &\vdots& \nonumber
\end{eqnarray}

Equation~\eqref{eq:Order1} corresponds to first order terms in $\epsilon$ and using the known expression of $\cK_0$ it reads:
$$
\cL_A(\cK_1(\rho_s))+ \Lint(\orho_A\otimes\rho_s)+\orho_A \otimes\cL_B(\rho_s) = \orho_A\otimes\cL_1(\rho_s)
.
$$
Since $\tr_A(\cL_A)=0$,  one  takes  the partial trace over $\cH_A$ of this equation and  gets
\begin{equation}
 \cL_1(\rho_s) = \tr_A\Big( \Lint(\orho_A\otimes\rho_s)\Big) + \cL_B(\rho_s) \label{eq:L1} \; .
\end{equation}
To obtain $\cK_1$, we must invert $\cL_A$. Remark~\ref{rmk:InvLindblad}, in appendix, shows that the general solution $\bX$ of $-\cL_A(\bX)= \bW-\tr(\bW)\orho_A$ can be written as
\begin{equation}\label{eq:InvLA}
\bX= \otau\coK_A(\bW) + \lambda \orho_A \; , \;\; \forall \lambda \in \mathbb{C} \; .
\end{equation}
where $\coK_A$ is a linear trace-preserving and complete positive map on $\cH_A$ and  $\otau$ a positive constant.

When $\ker\orho_A \subset \ker \bW$, i.e. $\bW=\tilde{\bW} \orho_A$ for some operator $\tilde{\bW}$,  lemma~\ref{lem:incluNoyau} shows that $\ker\orho_A \subset \ker \bX$, i.e. that  $\coK_A(\tilde{\bW} \orho_A) = \tilde{\bX} \orho_A$ for some operator $\tilde{\bX}$ (relevant when $\orho_A$ is not of full rank).   We will write as usual $\coK_A$ meaning $\coK_A \otimes \cI$ when applied to operators on $\cH_A \otimes \cH_B$. With this super-operator, using the just computed expression for $\cL_1(\rho_s)$, we readily get
\begin{eqnarray}
  \cK_1(\rho_s)&=& \otau \coK_A \Big(\Lint(\orho_A\otimes\rho_s)\Big) + \orho_A\otimes \bG_{1,B}(\rho_s)
    \label{eq:K1}
\end{eqnarray}
where $\bG_{1,B}$ is any Hermitian operator on $\cH_B$ (gauge degree of freedom).
We will show in the following how the gauge choice enables to select between system propertie.
We will show in the following how to fix the gauge in order to guarantee different properties (e.g. trace preservation, complete positivity, or a simple expression) of the expansion.

Lemma~\ref{lem:TrALint}, in appendix, proves that $\rho_s\mapsto \tr_A(\Lint(\orho_A\otimes \rho_s))$ remains of Lindblad type when $\Lint$ is of Lindblad type. Thus $\cL_1$ in \eqref{eq:L1} is of Lindblad type and $\dotex \rho_s = \epsilon \cL_1(\rho_s)$ can be interpreted in this case as a Zeno dynamics. For $\cK$, with the gauge $\bG_{1,B}=0$ in~\eqref{eq:K1}, we have
\begin{multline*}
 \cK(\rho_s)= \cK_0(\rho_s) + \epsilon \cK_1(\rho_s) + O(\epsilon^2) \\= \coK_A \left( \orho_A\otimes \rho_s + \epsilon\otau \Lint(\orho_A\otimes \rho_s)\right) + O(\epsilon^2)
 \\
  = \coK_A \left( e^{\epsilon\otau \Lint} (\cK_0(\rho_s)) \right) + O(\epsilon^2)
  .
\end{multline*}
Since $\coK_A$, $e^{\epsilon\otau \Lint}$ and $\cK_0$  are linear completely positive trace-preserving (CPTP) maps, the first order approximation $\cK_0+ \epsilon \cK_1$  coincides, up to second order terms, with the linear CPTP map $\coK_A \circ  e^{\epsilon\otau \Lint} \circ \cK_0$.
With the linear gauge $\bG_{1,B}(\rho_s)=-\otau \tr_A \Big(\Lint(\orho_A\otimes\rho_s)\Big)$, we get instead
\begin{equation}\label{eq:GenGaugeIFAC}
 \cK(\rho_s)= \coK_A \circ e^{\epsilon\otau \Lint}  \circ \cK_0 \circ e^{-\epsilon\otau \ra{\tr_A\Big( \Lint(\orho_A\otimes\rho_s)\Big)}}(\rho_s) + O(\epsilon^2).
\end{equation}
This expression is not always completely positive due to the backwards propagation with $\cL_1$, but it is e.g.~when $\cL_1$ is Hamiltonian (thus in particular when $\Lint$ is Hamiltonian).
\vspace{2mm}

The computations of the second order corrections $\cK_2$ and $\cL_2$ can be done via~\eqref{eq:Order2} along the same lines. The obtained expressions will depend on the gauge choice. We conjecture that, at any order $n$ versus $\epsilon$, we can choose  $(\cK_j,\cL_j)_{1\leq j \leq n}$  such that all equations corresponding to orders less that $n$ are satisfied and additionally, such that  $\sum_{j=0}^{n} \epsilon^j \cK_j(\rho_s)$ and  $\sum_{j=1}^{n} \epsilon^j \cL_j(\rho_s)$ coincide, up to $n+1$ order terms, with a trace-preserving completely positive map and with a Lindbladian dynamics, respectively. We have seen on several examples that the related expressions can be quite involved and this is the topic of ongoing research.
For instance, in general, the autonomous dynamics $\cL_B$ undesirably appears in $\cL_2$ if we take the gauge choice $\bG_{1,B}=0$ which always guarantees the map $(\cK_0+\epsilon \cK_1)$ to be CPTP.

\ra{We emphasize here that all gauge choices describe the same physical model i.e. for any gauge choice, the physical evolution of $\rho=\cK(\rs)$ (truncated at any order $\epsilon^k$) gives the same result (up to $\epsilon^{k+1}$ order terms). In this sense, the choice of the gauge degree of freedom has to be made in order to allow a better physical interpretation of $\rs$ and/or ensure physical properties (e.g. trace preservation, complete positivity $\dots$).}

Therefore in this paper, we have carried out further analysis leading to more explicit expressions up to second order, for particular structures that are relevant in typical quantum systems. The associated conclusions, e.g.~expressions proving that we retrieve trace-preserving completely positive maps and Lindblad type dynamics at this higher order as well, are presented in the next sections.


\section{Hamiltonian interaction} \label{sec:HamSyst}

One typical structure is when $\Lint$ reduces to a Hamiltonian interaction:
\begin{equation}
  \label{eq:Hinteraction}
  \dotex \rho = \cL_A(\rho) + \epsilon \Big( - i [ \Hint,\rho] + \cL_B(\rho) \Big)
\end{equation}
where $\Hint= \sum_{k=1}^{m} \bA_k \otimes \bB_k^\dagger$ where $\bA_k$ and $\bB_k$ are (not necessarily Hermitian) operators on $\cH_A$ and $\cH_B$, respectively. Such models are encountered e.g.~in reservoir engineering, where $\cH_A$ is the main dissipator allowing to stabilize the target system $\cH_B$ (see example below). More generally, weak Hamiltonian coupling of a target system $\cH_B$ to some environment is a standard model. The two typical cases leading to $\Hint$ Hermitian are $\bA_k$ and $\bB_k$ Hermitian, as appears in ``dispersive coupling'' ($\bA_k$ and $\bB_k$ are diagonal in the energy bases of the respective subsystems); and $\bA_{k+1} \otimes \bB_{k+1}^\dagger = \bA_k^\dagger \otimes \bB_k$, as appears in ``resonant coupling'' ($\bA_k$ and $\bB_k$ are annihilation operators). We will treat these two cases separately for the second-order approximation.\vspace{2mm}

For the first order approximation, we just plug this particular setting in the general formula and we get
\begin{eqnarray}
 \cL_1(\rho_s) & = &  - i \sum_{k=1}^m \, \left[ \tr(\bA_k \orho_A) \bB_k^\dagger\,,\;\rho_s\right] + \cL_B(\rho_s) \label{eq:L1c} \; .
\end{eqnarray}
In particular, $\cL_1$ is obviously of Lindblad type as the correction due to the coupling amounts to the Zeno Hamiltonian associated with $\Hint$ (i.e.~the Hamiltonian with $\cH_A$ frozen to $\orho_A$ by the fast dynamics).

For the remainder of this section, we take  the gauge $\bG_{1,B}(\rho_s)=-\otau \tr_A \Big(\Lint(\orho_A\otimes\rho_s)\Big)$ that will lead to simpler explicit expressions. First, for $\cK_0 + \epsilon\cK_1$, by directly plugging the Hamiltonian interaction into \eqref{eq:Order1} we obtain:
\begin{multline}\label{eq:Hint:K1}
\cK(\rho_s) = \cK_0(\rho_s) + \epsilon\cK_1(\rho_s) +O(\epsilon^2) =\\ \exp\big(\text{-}i \epsilon \bM \big) \; (\orho_A \otimes \rho_s) \; \exp\big(i \epsilon \bM^\dagger \big) \;\; +O(\epsilon^2) \quad
\text{ with } \bM = \sum_{k=1}^{m} \bF_k \otimes \bB^\dagger_k \; , 
\end{multline}
where $\bF_k \orho_A = \otau \coK_A \left(\, \bA_k \; \orho_A\, \right) - \otau \tr(\bA_k \; \orho_A) \orho_A$
satisfies $\tr(\bF_k \orho_A)
 = 0$, for $k=1,2,...,m$. Thus the abstract inversion formula now has to be applied only on the given interaction operators $\bA_k \orho_A$ ($k=1,2,...,m$), instead of on any possible operator. Note that expression \eqref{eq:Hint:K1} differs from the general one \eqref{eq:GenGaugeIFAC} by second order terms and, while the expression without the $O(\epsilon^2)$ terms is completely positive, $\bM$ is not always Hermitian so corrections of $O(\epsilon^2)$ might be necessary here to exactly preserve the trace.

We now turn to the second order.


\subsection{Second order approximation}

The gauge choice $\bG_{1,B}(\rho_s)=-\otau \tr_A \Big(\Lint(\orho_A\otimes\rho_s)\Big)$ also facilitates the expression of $\cL_2$. Taking the partial trace $\tr_A$ on \eqref{eq:Order2}, note that $\tr_A(\coK_A(\bX)) = \tr_A(\bX)$ for all $\bX$ on $\cH_A \otimes \cH_B$ thanks to trace preservation, and thus $\tr_A(\cK_1(\bX)) = \otau (\cL_1(\bX)-\cL_B(\bX)) + \bG_{1,B}(\bX)$ for all $\bX$ on $\cH_B$. From there we get simply:
\begin{eqnarray}
\cL_2(\rho_s) & = & \tr_A\Big(\Lint\Big(\cK_1(\rho_s \Big) \Big) \label{eq:Hint:L2} \; .
\label{eq:forL2}
\end{eqnarray}

\subsubsection{Dispersive interaction}
Let us first consider the case of a single interaction term, i.e.~where $\Hint = \bA \otimes \bB$ with $\bA,\bB$ both Hermitian. Plugging this expression of $\Lint$ and \eqref{eq:Hint:K1} for $\cK_1$ into \eqref{eq:forL2}, we obtain
$$\cL_2(\rho_s) = \tr(\bF \orho_A \bA + \bA \orho_A \bF^\dagger) \bB \rho_s \bB - \tr(\bA \orho_A \bF^\dagger)\rho_s \bB \bB - \tr(\bF \orho_A \bA) \bB \bB \rho_s \, .$$
By separating real and imaginary coefficients in the last two terms this can be rewritten as
\begin{equation} \label{eq:L2dispersive}
  \cL_2(\rho_s) = -i\left[\, \frac{\tr(\bF \orho_A \bA - \bA \orho_A \bF^\dagger)}{2i} (\bB)^2 ,\,\rho_s \,\right]  \;\; + \; \tr(\bF \orho_A \bA + \bA \orho_A \bF^\dagger) \, \cD_{\bB}(\rho_s) \; .
\end{equation}
The first term contains a Hermitian Hamiltonian, second-order correction to the Zeno Hamiltonian of $\cL_1$. The second term contains a dissipation in $\bB$, at the rate $\tr(\bF \orho_A \bA + \bA \orho_A \bF^\dagger)$. It expresses how the quantum uncertainty among eigenstates of $\bA$ on $\cH_A$, implies uncertainty in the Hamiltonian evolution with $\bB$ on $\cH_B$ when the two systems are dispersively coupled. There only remains to check that we will always get $\tr(\bF \orho_A \bA + \bA \orho_A \bF^\dagger) \geq 0$, which is done by Lemma \ref{Lem:Xipos} in appendix.

Thus, the second order reduced dynamics for the dispersive interaction is given by:
\begin{eqnarray}\label{eq:L1L2dispersive}
  \dotex \rho_s & = & -i \epsilon \left[\, \tr(\ba \orho_A)\bB, \, \rs \right] + \epsilon \cL_B(\rs) \nonumber \\
 & & -i\epsilon^2\left[\, \frac{\tr(\bF \orho_A \bA - \bA \orho_A \bF^\dagger)}{2i} (\bB)^2 ,\,\rho_s \,\right]  \;\;  + \epsilon^2 \tr(\bF \orho_A \bA + \bA \orho_A \bF^\dagger) \, \cD_{\bB}(\rho_s) \; . \;\;
\end{eqnarray}
with $\bF$ given by $\bF \orho_A = - \cL_A^{-1}\big(\bA- \tr(\bA \; \orho_A) \orho_A)\big)= \ra{+}\otau \coK_A \big(\, \bA \; \orho_A-\tr(\bA \; \orho_A) \orho_A \big)$.

\subsubsection{General case and resonant interaction}\label{ssec:GeneralLambda}

For $\Hint$ comprising several terms, we can still put $\cL_2$ in a form with Hamiltonian and dissipation super-operator by following the same procedure. Instead of directly getting a dissipation rate $\tr(\bF \orho_A \bA + \bA \orho_A \bF^\dagger)$ as for the dispersive interaction, we here get a positive semidefinite matrix in $\mathbb{C}^{m \times m}$ which determines dissipation operators consisting of linear combinations of the $\bB_k$. The number of decoherence channels in $\cL_2$ will be equal to the number $\leq m$ of linearly independent terms in $\Hint$. Explicitly, plugging the expression \eqref{eq:Hint:K1} of $\cK_1$ and the Hamiltonian interaction into \eqref{eq:Hint:L2}, a few algebraic manipulations readily yield:
\begin{eqnarray}
\nonumber \cL_2(\rho_s) & = &
- i \left[\; \sum_{k,j}  Y_{k,j}\, \bB_k \bB_j^\dagger \; , \; \rho_s \right] + \sum_{k,j}  X_{k,j} \, \left(\,
\bB_j^\dagger \rho_s \bB_k - \tfrac{1}{2}\left( \bB_k \bB_j^\dagger \rho_s + \rho_s \bB_k \bB_j^\dagger \right)
\,\right) \\
\label{eq:L2bigfirst} {\text{with}} & & X_{k,j}= \tr\left( \bF_j \orho_A \bA_k^\dagger  + \bA_j \orho_A \bF_k^\dagger \right)\\
\nonumber & & Y_{k,j} = \frac{1}{2i}\, \tr\left( \bF_j \orho_A \bA_k^\dagger  - \bA_j \orho_A \bF_k^\dagger \right)  \;\;\; \text{ for } j,k=1,2,...,m \, .
\end{eqnarray}
Both matrices $X$ and $Y$ $\in \mathbb{C}^{m \times m}$ are Hermitian. Thus the first term contains a Hermitian Hamiltonian. Lemma \ref{Lem:Xipos} in appendix proves that $X$ is positive semidefinite. Then equation \eqref{eq:L2bigfirst} is of standard Lindblad form. Indeed, writing $X = \Lambda \Lambda^\dagger$, where $\Lambda$ can be obtained for instance easily by Cholesky factorization, we get
$$\cL_2(\rho_s) = - i \left[\; \sum_{k,j}  Y_{k,j}\, \bB_k \bB_j^\dagger \; , \; \rho_s \right] + \sum_{k=1}^m \cD_{\bL_k}(\rho_s) \;\; \text{ where } \;  \bL_k = \sum_{j=1}^m \ra{\Lambda_{j,k}^* \bB_j^\dag} \;\;, \;\;\;\; X = \Lambda \Lambda^\dagger  \; .
$$
The second order reduced dynamics is then given by:
\begin{equation}\label{eq:L1L2Hamiltonnian}
  \dotex \rho_s = - i \left[\;\epsilon\sum_k \tr(\bA_k \orho_A) \bB_k +  \epsilon^2 \sum_{k,j}  Y_{k,j}\, \bB_k \bB_j^\dagger \; , \; \rho_s \right] + \epsilon \cL_B(\rs) + \epsilon^2 \sum_{k=1}^m \cD_{\bL_k}(\rho_s) \, .
\end{equation}
The choice of $\Lambda$ is not unique, reflecting the non-uniqueness of Lindblad representations with several dissipation channels. Usually and as illustrated in the example of Section \ref{ssec:Ex2photon}, a preferred decomposition, identifying physically meaningful components, can be imposed.
\vspace{2mm}

For a single resonant interaction $\Hint = \bA \otimes \bB^\dagger + \bA^\dagger \otimes \bB$, we thus have $\bF_1$ and $\bF_2$ given essentially by $-\cL_A^{-1}\equiv \otau\coK_A$ acting on $\bA \orho_A- \tr{\bA\orho_A}\orho_A$ and on $\bA^\dagger \orho_A- \tr{\bA^\dag \orho_A}\orho_A$ respectively; $X,Y,\Lambda$ $\in \mathbb{C}^{2 \times 2}$ defined as above with no particular simplifications, unless $\orho_A$ takes a particular form; and finally
\begin{eqnarray}
\tfrac{d}{dt} \rho_s & = & - i \, \epsilon \left[ \tr(\bA \orho_A) \bB^\dagger + \tr(\bA^\dagger \orho_A) \bB\,,\;\rho_s\right] \nonumber \\
& & - i \, \epsilon^2 \left[ Y_{1,1} \bB \bB^\dagger + Y_{2,2} \bB^\dagger \bB + Y_{1,2} \bB \bB + Y_{2,1} \bB^\dagger \bB^\dagger \; , \;\rho_s\right] \nonumber\\
& & + \epsilon \, \cL_B(\rho_s) \nonumber\\
& & + \epsilon^2 \cD_{(\Lambda_{1,1} \bB^\dagger + \Lambda_{2,1} \bB)}(\rho_s) + \epsilon^2 \cD_{(\Lambda_{1,2} \bB^\dagger + \Lambda_{2,2} \bB)}(\rho_s) \; .
\label{eq:L1L2resonnant}
\end{eqnarray}
\vspace{2mm}

The following example considers a combination of dispersive and resonant interaction, where $m=3$. The result on this example is new to our knowledge, and with the number of decoherence channels that have to be combined, it is not easily obtained from intuitive reasoning.

\subsection{Example: Two-photon pumping}\label{ssec:Ex2photon}

We illustrate the result of Section \ref{ssec:GeneralLambda} on the model of the experiment presented in \cite{theseJoachimCohen,TouzardOngoing}. Such system, following the theoretical proposal \cite{mirrahimi2014dynamically}, is a promising way towards dynamically protected quantum processors. The reservoir is based on a scheme that induces 2-photon loss at order $\epsilon^2$ on the target system, as a dominant effect. We show how our method can calculate a more precise reduced model taking into account \emph{all} the effects at order $\epsilon^2$, and thereby quantify the effects of potential 2-photon excitation and cross-Kerr nonlinearity on this reservoir.

The system is composed of two interacting cavities: the first cavity (fast subsystem, $\cH_A$) is driven by an electromagnetic field and exchanges energy with the environment, with an energy loss term dominant with respect to both the energy gain term and the electromagnetic field drive. This implies that only the  lowest energy level is  significantly populated, so we can model the cavity as a two level system, i.e.~a qubit with energy levels $\ket{g},\ket{e}$.
The second cavity (slow subsystem, $\cH_B$) weakly interacts with this qubit, with auxiliary pump fields matched such that the resonant interaction exchanges 1 energy quantum of $\cH_A$ with 2 energy quanta of $\cH_B$ \cite{theseJoachimCohen,TouzardOngoing}; an additional residual dispersive interaction (``cross-Kerr'') is unavoidable in this setup.

The system model thus writes as follows. As usual, $\bSx,\bSy,\bSz$ denote the standard Pauli matrices,
$\bSm=|g\rangle \langle e|$ and $\bSp=|e\rangle \langle g|$ the energy loss operator and energy gain operators respectively, all on $\cH_A$, and $\bb,\bb^\dag$ denote respectively the annihilation
and creation operators, on $\cH_B$.
The dynamics of the fast subsystem are described by the Lindbladian operator
\begin{equation} \label{lind A}
\mathcal{L}_A = - i u [ \bSy,\cdot] +\kappa_- \mathcal{D}_{\bSm}+\kappa_+ \mathcal{D}_{\bSp},
\end{equation}
where the coupling constants satisfy
\begin{equation} \label{orders 2 ph}
\frac{\kappa_+}{\kappa_-}\ll 1 \quad , \quad\frac{|u|}{\kappa_-}\ll 1.
\end{equation}
The interaction of the two systems is described by
\begin{eqnarray} \label{ham}
\Hint & = & g\bSp \otimes \bb^2 + g\bSm
\otimes\bb^{\dag 2} + \chi |e\rangle \langle e| \otimes \bb^\dag \bb \\
\nonumber & = & \bA_1 \otimes \bB_1^\dagger + \bA_2 \otimes \bB_2^\dagger + \bA_3 \otimes \bB_3^\dagger \, .
\end{eqnarray}
The fact that $\Hint$ is much weaker than $\mathcal{L}_A$ is expressed by $|g|,|\chi|\ll \kappa_+,u,\kappa_-$. Formally we can take units such that $\kappa_+,|u|,\kappa_-$ are of order $1$ or larger, and $g,\chi$ are of order $\epsilon \ll 1$. In particular, if $g=\chi=0$, then the two systems are independent. The fast system converges to the steady state
\newcommand{\rrba}{\bar{\rho}_A}
\begin{align*}
\rrba &= \frac{\bI + x_\infty \bSx + z_\infty \bSz}{2} \\
\text{with } x_\infty &= \frac{4 u(\kappa_+-\kappa_-) }{(\kappa_++\kappa_-)^2+8u^2} \;\; \text{ and } \;\;
z_\infty = \frac{\kappa_+^2-\kappa_-^2 }{(\kappa_++\kappa_-)^2+8u^2},
\end{align*}
while the slow target system $B$ does not move.

When $g,\chi$ are nonzero, the dynamics of the slow dynamics, approximately corresponding to the second cavity $\cH_B$, can be written in the form \eqref{eq:Lexpansion}. By simple computations this yields the
first-order dynamics of system $B$, given by the Zeno Hamiltonian:
\[
\epsilon \mathcal{L}_1(\rho_s)=-i\left[\chi \frac{1+z_{\infty}}{2}\bb^{\dag}\bb + g x_{\infty} \frac{\bb^{\dag 2}+\bb^2}{2},\; \rho_s \; \right].
\]

We now compute the second-order dynamics, choosing in equation \eqref{eq:InvLA} the gauge $\bG_{1,B}(\rho_s)=-\otau \tr_A \Big(\Lint(\orho_A\otimes\rho_s)\Big)$.
Algebraic computations, by solving for $\cL_A^{-1}$ with the Bloch equations for the qubit, yield
$X$ and $Y$ in the form:
$$\begin{array}{lllll}
\epsilon^2 \, X_{1,1}=&  \left(\frac{z_\infty}{2}(3x_\infty^2-2) - \frac{x_\infty^2}{2}+z_\infty^2-(z_\infty-1)\frac{\kappa_--\kappa_+}{\kappa_-+\kappa_+}\right)\; g^2 / (\kappa_--\kappa_+)\\
\epsilon^2 \, X_{2,2}=& \left(\frac{z_\infty}{2}(3x_\infty^2-2) + \frac{x_\infty^2}{2}-z_\infty^2+(z_\infty+1)\frac{\kappa_--\kappa_+}{\kappa_-+\kappa_+}\right)\; g^2 / (\kappa_--\kappa_+)\\
\epsilon^2 \, X_{3,3}=& \frac{z_\infty}{2}\left(z_\infty^2-x_\infty^2-1\right)\; \chi^2/ (\kappa_--\kappa_+)\\
\epsilon^2 \, X_{1,2}=& \left(\frac{z_\infty}{2}(3x_\infty^2-2) - \frac{\kappa_--\kappa_+}{\kappa_-+\kappa_+}\right) \; g^2   / (\kappa_--\kappa_+) \\
\epsilon^2 \, X_{1,3}=& x_\infty \left(z_\infty^2 - \frac{x_\infty^2}{4}-\frac{z_\infty}{2}+\frac{1}{2}\frac{\kappa_--\kappa_+}{\kappa_-+\kappa_+}\right)\; \chi g / (\kappa_--\kappa_+)\\
\epsilon^2 \, X_{2,3}=& x_\infty \left(z_\infty^2 - \frac{x_\infty^2}{4}+\frac{z_\infty}{2}-\frac{1}{2}\frac{\kappa_--\kappa_+}{\kappa_-+\kappa_+}\right) \; \chi g / (\kappa_--\kappa_+) \;\;;
\end{array}$$
\begin{align*}
\begin{array}{lll}
\epsilon^2 \, Y_{1,2} &= -\left(2z_\infty^2-x_\infty+ 2z_\infty \frac{\kappa_- - \kappa_+}{\kappa_- + \kappa_+} \right)\; g^2 / (4i(\kappa_- - \kappa_+))  \\
\epsilon^2 \, Y_{1,3} &=  \left(x_\infty - x_\infty z_\infty - x_\infty^3/2 - z_\infty^2 x_\infty - x_\infty \frac{\kappa_- - \kappa_+}{\kappa_- + \kappa_+} \right)\; g\chi / (4i(\kappa_- - \kappa_+)) \\
\epsilon^2 \, Y_{2,3} &=   \left(x_\infty + x_\infty z_\infty - x_\infty^3/2 - z_\infty^2 x_\infty + x_\infty \frac{\kappa_- - \kappa_+}{\kappa_- + \kappa_+} \right)\;  g\chi / (4i(\kappa_- - \kappa_+)) \; .
\end{array}
\end{align*}

Towards interpreting these expressions, we take into account the relative strengths of the couplings \eqref{orders 2 ph}. More precisely, with $\kappa_-=1+n_{\text{thermal}}$, $\kappa_+=n_{\text{thermal}}$
and $n_{\text{thermal}}\ll 1$, we define $\delta^2=\kappa_+/\kappa_- \approx n_{\text{thermal}}$ and $\eta=u/\kappa_-$, and we neglect the terms of order three or higher in $\delta,\eta$.

The Hamiltonian operator in \eqref{eq:L1L2Hamiltonnian} then reads
\[
\epsilon^2 \sum_{k,j}  Y_{k,j}\, \bB_k \bB_j^\dagger
\;\;=\;\; \frac{2i g \chi}{\kappa_-}\,  \eta\, \big(\bb^{\dag 3}\bb - \bb^\dag \bb^{ 3} \big) + \frac{8i g^2}{\kappa_-}\, \eta^2\, \big(\bb^{\dag 4}-\bb^4\big) 
 \]
up to terms of order at least three in $\delta,\eta$.
For the dissipative part, up to the same terms,
\[
\epsilon^2 X \;\; = \;\; \frac{1}{\kappa_-}
\begin{pmatrix}
\left(4 - 8 \delta^2 - 64 \eta^2\right) g^2& -32 \eta^2 g^2 &
  -8 \eta g \chi \\
 -32 \eta^2 g^2 &  4 \delta^2 g^2  &  0\\
 -8  \eta g \chi & 0 &
  \left(2 \delta^2 + 16 \eta^2\right)\chi^2
\end{pmatrix} \; .
\]
From this form, it is already clear that the dominant effect of the dissipation involves the two-photon annihilation operator. We next write $X=\Lambda \Lambda^{\dag}$, choosing $\Lambda^{\dag}$ as an upper triangular matrix:
\[
\epsilon \Lambda^{\dag}=
\frac{1}{\sqrt{\kappa_-}}
\begin{pmatrix}
2(1-\delta^2-8\eta^2) g & -16 \eta^2 g & -4 \eta \chi \\
0 & 2 \delta g & 0 \\
0 & 0 & \sqrt{2}\delta \chi
\end{pmatrix}
\]
so that the three dissipation channels are given by the operators
\begin{align*}
\begin{array}{lll}
\epsilon \bL_1&=\frac{1}{\sqrt{\kappa_-}}\, \big(\,2(1\text{-}\delta^2\text{-}8\eta^2)\, g\bb^2 -4 \eta\, \chi \bb^{\dag}\bb - 16 \eta^2\, g\bb^{\dag 2}  \,\big) \\
\epsilon \bL_2&=\frac{1}{\sqrt{\kappa_-}}\, 2\delta\, g \bb^{\dag 2}\\
\epsilon \bL_3&=\frac{1}{\sqrt{\kappa_-}}\, \sqrt{2}\delta\, \chi\bb^{\dag}\bb
\end{array}
\end{align*}
With the same approximation, the terms in $\cL_1$ involve $(1+z_{\infty})/2 = \delta^2 + 4 \eta^2$ and $x_{\infty}/2 = -2 \eta$. Wrapping up, the effects on the slow subsystem are thus, with units such that $\kappa_-=1$:
\begin{itemize}
\item At order $(g \eta)$ and $(g^2)$ respectively: a two-photon pumping Hamiltonian, in $\bb^{\dag 2}+\bb^2$, and a two-photon dissipation, with $\bL_1 \approx \bb^2$, precisely as intended in \cite{mirrahimi2014dynamically};
\item At order $(\chi \delta^2,\chi \eta^2)$: a Stark shift Hamiltonian $\bb^\dag \bb$, which just shifts the cavity frequency and can be compensated for;
\item At order $(g\chi \eta)$: a Hamiltonian in $(\bb^\dag \bb^{ 3} - \bb^{\dag 3}\bb)/i$, whose precise deformation effect would have to be investigated; and a modification of the two-photon dissipation channel $\bL_1$ by some dephasing effect, leading to terms like $\bb^2 \rho \bb^\dag \bb$.
\item At order $(\chi^2 \delta^2, g^2 \delta^2, g^2 \eta^2)$: a Hamiltonian effect in $(\bb^4-\bb^{\dag 4})/i$, that is essentially a 4-photon drive; two new dissipation channels, namely in $\bb^{\dag 2}$ and $\bb^\dag \bb$; and a further modification of $\bL_1$, now leading to terms like $\bb^2 \rho \bb^2$.
\end{itemize}
These are the dominant ones provided $|g|,|\chi| \ll |\eta|,|\delta| \ll 1$; once $g,\chi$ become comparable to $\eta,\delta$, effects of order $\epsilon^3$ might become important as well. The next step of our asymptotic expansion method would address this point; carrying out the related computations goes beyond the present paper.

Finally, we must emphasize that the reduced dynamics describes the evolution of the \emph{slow subsystem $\rho_S$. This $\rho_S$ does not exactly coincide with the state of subsystem $B$.} Indeed, the Kraus map $\cK_0+\epsilon\cK_1+...$ (not given here) expresses precisely how subsystem $B$ gets slightly hybridized with subsystem $A$ in order to obtain Markovian slow dynamics. The dynamics of just $\tr_A(\rho)$ would not be Markovian at this order of precision, and this is one reason that precludes a simple intuitive derivation of these results.


\section{Cascade interaction} \label{sec:CascSyst}
We consider in this section a different kind of interaction between the fast and the slow subsystem. Instead of the standard reciprocal Hamiltonian coupling of Section \ref{sec:HamSyst} we consider a unidirectional coupling which allows the output of the first system to feed the second system while forbidding the reverse process. The expression "cascaded system" was first introduced in \cite{CarmichaelPRL93}, see also e.g. \cite{GardinerZollerBook} for more details on this type of interaction. \ra{The network extension of the cascade structure can be found in \cite{GoughTAC09}, see \cite{GoughJMP10, NurdinPTRSA12} for its relation with adiabatic elimination.}

The master equation of a fast subsystem $A$ on a Hilbert space $\mathcal{H}_A$ coupled via cascaded interactions to a slow subsystem $B$ on a Hilbert space $\mathcal{H}_B$ is given by:
\begin{eqnarray}
\dotex \rho = \cL_A(\rho) +   \cD_{\ba + \epsilon\bb}(\rho) +\frac{\epsilon}{2}  [\ba^\dag  \bb- \ba \bb^\dag,\rho] + \epsilon^2 \cL_B(\rho)
 \nonumber
\\\qquad = \cL_A(\rho) + \cD_{\ba} (\rho) + \epsilon \Big(\ba [\rho,\bb^\dag] +  [\bb,\rho] \ba^\dag\Big) + \epsilon^2 \Big(\cD_{\bb}(\rho) + \cL_B(\rho) \Big)
\label{eq:Cinteraction}
\end{eqnarray}
where $\cL_A$ and $\cL_B$ are Lindbladian  super-operators acting only on $\cH_A$ and $\cH_B$, respectively, and where
 $\ba$ and $\bb$  are operators acting only on  $\cH_A$ and $\cH_B$ respectively. We assume that there exists a unique steady state density operator $\orho_A$ solution of $\cL_A(\rho) + \cD_{\ba} (\rho) = 0$. Note that unlike the Hamiltonian interaction case, the fast dynamics is given by the super-operator $\cL_A(\rho)$ but also by $\cD_{\ba}(\rho)$  due to the unidirectional coupling. The completely positive map $\coK_A$ of Remark~\ref{rmk:InvLindblad} is defined accordingly.

 We detail now the computation leading to the adiabatic elimination of the fast variables associated with the subsystem $A$ in the case of a cascaded system whose dynamics is given by \eqref{eq:Cinteraction}. The zero order approximation is readily given by $\cL_0(\rs) = 0$ and $\cK_0(\rs) = \orho_A \otimes \rs$.

The first order reduced dynamics can be directly computed using equation \eqref{eq:L1}. A straightforward calculation yields:
\begin{equation} \label{eq:CascLindOrd1}
\cL_1 (\rs) = \left[ \tr(\orho_A \ba^\dag) \bb - \tr(\ba \orho_A) \bb^\dag, \rs \right]
\end{equation}
The super-operator $\cK_1$ corresponding to the first order entanglement between the two subsystems is given by \eqref{eq:K1}. From this expression, with the gauge choice $\bG_{1,B}(\rs)= - \otau \tr_A \bigl(\cL_{int}(\orho_A \otimes \rs) \bigr)$ we get
\begin{equation}
\cK_1(\rs) = \otau \coK_A \left( \oa \hspace*{1mm} \orho_A \right) \otimes \left(\rs \bb^\dag - \bb^\dag \rs \right) + \otau \coK_A \left( \orho_A \oa^\dag \right) \otimes \left(\bb \rs - \rs \bb \right)
\end{equation}
where $\oa = \ba - \tr(\ba \orho_A)\bI_A$. We made this particular gauge choice in order to ensure $\tr_A(\cK(\rs))= \rs$ and therefore $\rs$ corresponds to the density operator of the subsystem $B$. This   gauge choice is different from the one of   Section \ref{sec:BipartSyst}. However   we show that  $\cK(\rs)=\cK_0(\rs) + \epsilon \cK_1(\rs)$ is also a completely positive trace-preserving map  up-to second order terms.  To exhibit such property, we express this map in an explicit Kraus operator form.

Using Lemma \ref{lem:incluNoyau} and Corollary \ref{cor:inclu_noyau} we get that $\coK_A \left( \oa \hspace*{1mm} \orho_A \right) = \coK_A \left( \oa \hspace*{1mm} \orho_A \right)\orho_A^{-1} \orho_A$ and $\coK_A \left( \orho_A \oa^\dag \right)= \coK_A \left( \orho_A \oa^\dag \right)\orho_A^{-1} \orho_A$ where $ \orho_A^{-1}$ stands for the Moore-Penrose pseudo-inverse of $\orho_A$. Therefore we derive the following Kraus operator form for the first order entanglement:
\begin{align}
\rho &=\cK_0(\rs) + \epsilon \cK_1(\rs) + O(\epsilon^2) = \bM (\orho_A \otimes \rs) \bM^\dag \\
\bM &= \biggl( I + \epsilon \otau \bigl( \coK_A \left( \orho_A \oa^\dag \right) \orho_A^{-1} \otimes b - \coK_A \left(\oa \, \orho_A  \right)\orho_A^{-1} \otimes \bb^\dag \bigr) \biggr)
\end{align}

 \subsection{Second order approximation}

 In order to compute the second order dynamics we proceed in the same manner as the Hamiltonian interaction case: first, we choose the gauge $\bG_{1,B}(\rs)= - \otau \tr_A \bigl(\cL_{int}(\orho_A \otimes \rs) \bigr)$ in order to facilitate the expression of $\cL_2$. Then, take the partial trace $\tr_A$  on equation \eqref{eq:Order2} leading to the cancellation of the unknown term $\cL_A(\cK_2(\rs))$. Using $\tr_A(\cL_A(\bullet)) \equiv 0$ with $ \tr_A(\cK_0(\bullet)) \equiv \bullet$ and the fact that $\tr_A ( \cK_1(\bullet)) \equiv 0$ due to our choice of the gauge degree of freedom we get:
\begin{equation}
\cL_2(\rs) =\tr_A \Big(\ba [\bb^\dag,\cK_1(\rs)] +  [\cK_1(\rs),\bb] \ba^\dag\Big) +  \bigl(\cD_{\bb} (\rs) + \cL_B(\rs) \bigr)
\end{equation}
Note that the term $\cD_{\bb}$ is intrinsically related to the cascaded structure while $\cL_B$ represents the slow dynamics on the subsystem $B$
and is independent of the structure (in particular, it could satisfy $\cL_B \equiv 0$).
Therefore the term $\cL_B$ must be treated independently from the other terms. For this reason we impose $\cL_2 = \widetilde{\cL_2} + \cL_B$ with
\begin{equation}
\widetilde{\cL_2} = \tr_A\Big(\ba [\bb^\dag,\cK_1(\rs)] +  [\cK_1(\rs),\bb] \ba^\dag\Big) + \cD_{\bb}
\end{equation}
A direct expansion leads to the following expression:
\begin{align}
\widetilde{\cL_2}(\rs)= \cD_{x_1 \bb + y_1 \bb^\dag}(\rs) + \cD_{x_2 \bb + y_2 \bb^\dag}(\rs) \ra{+\frac{\alpha^*-\alpha}{2}\left[\bb^\dag \bb - \bb \bb^\dag \, , \, \rs \right]} \label{eq:CascLindOrd2}
\end{align}
where $\{x_1, x_2, y_1, y_2 \} \in \mathbb{C}^4$ are solutions of the set of equations:
\begin{align}\left\{
\begin{array}{r c l}
|x_1|^2 + |x_2|^2 & = & \alpha+\alpha^*+1 \\
|y_1|^2 + |y_2|^2 & = & \alpha+\alpha^* \\
x_1 y_1^* + x_2 y_2^* & = & -2 \beta
\end{array}\right. \label{eq:CascSystCond}
\end{align}
with $\alpha = \otau \tr \left(\ba   \coK_A \left( \orho_A \oa^\dag \right) \right)$ and  $ \beta = \otau \tr (\ba^\dag \coK_A(\orho_A \oa^\dag))$.

This set of equation can be viewed as the norm and the scalar product between the two vectors $\{x_1, x_2\}^\intercal \in \mathbb{C}^2$ and $\{y_1, y_2\}^\intercal \in \mathbb{C}^2$. Therefore, there exists a (non-unique) solution to \eqref{eq:CascSystCond} if and only if holds:
\begin{align}
\alpha+\alpha^* &\geq 0  \label{eq:CascCond1} \\
(\alpha+\alpha^*+1)(\alpha+\alpha^*) &\geq 4|\beta|^2 \label{eq:CascCond2}
\end{align}
Using a slight adaptation of Lemma \ref{Lem:Xipos}, we get that condition \eqref{eq:CascCond1} holds for any system described by the master equation \eqref{eq:Cinteraction}. We conjecture that \eqref{eq:CascCond2} holds for any system. The non-uniqueness of the solutions is expected due to the fact that the Lindblad decomposition \eqref{eq:CascLindOrd2} is also not unique.

In the case of cascade interaction, the second order reduced dynamics is given by:
\begin{equation}\label{eq:L1L2Cascade}
  \dotex \rho_s = \epsilon \left[\, \tr(\orho_A \ba^\dag) \bb - \tr(\ba \orho_A) \bb^\dag, \rs \, \right] + \epsilon^2 \biggl( \cL_B(\rs) +  \cD_{x_1 \bb + y_1 \bb^\dag}(\rs) + \cD_{x_2 \bb + y_2 \bb^\dag}(\rs) \biggr)
\end{equation}

 \subsection{Arbitrary system with a squeezed drive} \label{subsec:Squeezed}

To illustrate our results, we consider a driven linear cavity  producing a  squeezed  output field, unidirectionally "feeding" the slow subsystem. Such kind system was first studied in \cite{GardinerPRL86} considering the slow subsystem as a qubit and shows how one can engineer the coherence times of this qubit using a squeezed  input field. It has recently been realized experimentally in \cite{SiddiqiNat13} validating the theoretical results. Our aim is to show through this example, how to apply our method, emphasizing that it doesn't need the specification of the slow subsystem and readily retrieve the known results of the adiabatic elimination of the linear cavity when the slow subsystem is a qubit.

 Using \cite{GoughTAC09} , we get the following dynamics for the system:
\begin{eqnarray} \label{eq:ExampleCasc}
\frac{d}{dt}\rho = g[\ba^2 - {\ba^\dag}^2, \rho] + \kappa \cD_{\ba}(\rho) + \epsilon \sqrt{\kappa}  \biggl( \ba[\rho, \bb^\dag]+[\bb,\rho]\ba^\dag \biggr) + \epsilon^2 \cD_{\bb}(\rho \bigr)
 \end{eqnarray}
where $\ba$ is the annihilation operator on the cavity, $\bb$ is an operator on the unspecified slow subsystem (usually the annihilation operator for a cavity and the energy loss operator for a qubit). The term $g[\ba^2 - {\ba^\dag}^2, \rho]$ corresponds to the standard squeezing Hamiltonian by an appropriate phase choice for $\ba$ ($g$ is real here). The parameter  $\kappa$ is the dissipation rate of the cavity. The time-scale separation is given by $\kappa \gg \epsilon^2$. This model is valid under the assumption $\kappa > 4g$, otherwise the cavity subsystem is unstable and its energy grows to infinity (and therefore additional phenomenon such as Kerr effect have to be taken into account). As any dynamics $\cL_B$ on the slow subsystem plays no role for second order   computations, we assume for simplicity $\cL_B \equiv 0$.

To derive the reduced dynamics up to second  order on this example, we use formulas \eqref{eq:CascLindOrd1} and \eqref{eq:CascLindOrd2} where the fast system $A$ is the cavity and the slow system $B$ is unspecified. Therefore we have only to compute the three coefficients $\tr(\sqrt{\kappa}\ba \orho_A)$, $\alpha = \kappa\otau \tr \left(\ba   \coK_A \left( \orho_A \oa^\dag \right) \right)$ and  $ \beta = \kappa\otau \tr (\ba^\dag \coK_A(\orho_A \oa^\dag)$. As the fast dynamics of \eqref{eq:ExampleCasc} corresponds to the evolution of a linear quantum harmonic oscillator, these coefficients are easier to compute by taking their Heisenberg representation counterpart i.e:
\begin{align*}
\tr(\ba \orho_A) = \tr\left(\ba ~e^{\infty \cL_A}(\rho_0)\right)=  \tr(e^{\infty \cL_A^*}(\ba) ~\rho_0) \\
\alpha = \kappa \tr \left(e^{\infty \cL_A^*}\left( \oa^\dag \int_0^\infty e^{t \cL_A^*} (\ba)~dt \right)   \rho_0 \right) \\
\beta = \kappa  \tr \biggl( e^{\infty \cL_A^*} \left( \oa^\dag \int_0^\infty  e^{t \cL_A^*} (\ba^\dag)   dt \right)  \rho_0 \biggr)
\end{align*}
where $\rho_0$ is the initial state, ${\cL_A}^*$ is the dual of $\cL_A$. With a small abuse of notation, $e^{\infty {\cL_A}^*} $ means $\lim_{t\to \infty} e^{t {\cL_A}^*}$.

Then, some usual  computations lead to (see Appendix \ref{apx:CalcSqueeze} for the key elements of the computation):
\begin{align*}
\tr(e^{\infty \cL_A^*}(\ba) ~\rho_0) = 0 \\
\alpha = \frac{32\kappa^2 g^2}{((\kappa+4g)(\kappa-4g))^2} \\
\beta = - \frac{64 g^3 \kappa + 4g \kappa^3}{((\kappa+4g)(\kappa-4g))^2}
\end{align*}
In this case one can check that $(\alpha+\alpha^*+1)(\alpha+\alpha^*) = 4|\beta|^2$  verifies condition \eqref{eq:CascCond2}. It is therefore possible to solve \eqref{eq:CascSystCond} and the simplest solution is given by
$ x_2 = y_2 = 0, x_1 = \sqrt{2\alpha+1}, y_1 = -\sqrt{2\alpha} $. Noting that $\alpha \in \mathbb{R}$, we get the second order reduced dynamics:
\begin{equation}
\frac{d}{dt}\rho_s = \epsilon^2 \cD_{\sqrt{2\alpha+1}\bb - \sqrt{2\alpha}\bb^\dag}(\rs)
\end{equation}
Recovering the result of  \cite{GardinerPRL86} when the slow subsystem is a qubit and $\bb = \bSm$ the energy loss operator.

\section{Conclusion}

We have presented an adiabatic elimination technique for bipartite quantum systems characterized by one subsystem converging rapidly towards a unique steady state. By considering the slow dynamics as a perturbation of order $\epsilon$, we explicitly obtain, up to second order terms in $\epsilon$, the reduced slow dynamics as a Lindblad master equation. Our explicit formulas directly apply to reservoir engineering settings e.g.~for computing perturbations on the engineered reservoir, or the residual dissipation added when first-order Hamiltonian Zeno dynamics is introduced towards performing gates on top of the reservoir.

The asymptotic expansion method, developed in this paper on bipartite systems, seems in fact to be applicable for any type of quantum system with two time-scales.  In~\cite{AzouitCDC15,theseJoachimCohen} for instance it is exploited in the context of fast convergence towards a decoherence free space, with some preliminaries results that would have to be completed. In this case, the fast dissipative dynamics drives the system towards a particular subspace inside the Hilbert space. Hence the slow/fast decomposition involves a Cartesian product rather than a tensor product as presented here.

We conjecture that this geometric adiabatic elimination technique will lead to formulas conveying a physical interpretation of the reduced model at higher orders as well: by ensuring the Lindblad form of the reduced dynamics, and by providing a completely-positive trace-preserving mapping from reduced model to actual system states. \ra{The key element to do so would be to properly choose the gauge degree of freedom at each order. } This is the topic of ongoing research.

\appendix

\section{Inverse of Lindblad super-operators  via Kraus maps}

Consider the Lindblad master equation $\dotex \rho =  \cL(\rho)$ where $\rho$ is a density operator on  a finite dimensional Hilbert space $\cH$.  Assume that for any operator $\bX$ on $\cH$,
$e^{t\cL} (\bX)$ converges exponentially towards a fixed point depending on $\bX$. This means that there exists a  complete-positive and trace-preserving map   $\coR$ such that
$\lim_{t\mapsto +\infty}   e^{t\cL} (\bX) = \coR(\bX)$.  Thus we   have $\cL(\coR(\bX))\equiv 0\equiv \coR(\cL(\bX))\equiv 0$ since $ e^{t\cL} (\coR(\bX))\equiv \coR(\bX) \equiv \coR(e^{t\cL}(\bX))$.

\begin{lemma} \label{lem:InvLindblad}
There exists $\bar{\tau}>0$ such that  the super-operator  $\coK$ sending  operator $\bX$ to
$$
\coK(\bX) = \frac{1}{\bar\tau} \int_0^{+\infty}  e^{t\cL}\Big(\bX- \coR(\bX)\Big)~dt + \coR(\bX)
$$
is a linear, trace-preserving and completely positive mapping with $$- \cL\big( \bar\tau\coK(\bX)\big) = \bX- \coR(\bX).$$
\end{lemma}

\begin{remark} \label{rmk:InvLindblad}
  When  for any initial density operator the solution of  $\dotex \rho =  \cL(\rho)$  converges toward a unique density operator $\overline{\rho}$, we have $\coR(\bX)=\tr(\bX) \overline{\rho}$. In this case,  for any given operator $\bW$ with $\tr(\bW)=0$,  the general solution of  $-\cL(\bX) = \bW$   reads
 $$\bX= \int_0^{+\infty}  e^{t\cL}(\bW)~dt + \lambda \overline{\rho} = \bar\tau \coK(\bW) + \lambda \overline{\rho}$$ where $\lambda$ is an arbitrary complex number. Moreover $\bX=\bar\tau \coK(\bW)$ is the unique solution  with zero trace.
\end{remark}

\begin{proof}
Due to exponential convergence of  $e^{t\cL} (\bX)$  towards $\coR(\bX)$, the indefinite integral
$$
\coM(\bX)\triangleq \int_0^{+\infty}  e^{t\cL}\Big(\bX- \coR(\bX)\Big)~dt
$$
is absolutely convergent. Since  $\dotex  e^{t\cL}\Big(\bX- \coR(\bX)\Big)= \cL\left(e^{t\cL}\Big(\bX- \coR(\bX)\Big)\right)$, we have $\cL(\coK(\bX)) = -\frac{\bX- \coR(\bX)}{\bar\tau}$.    Since, for each $t \geq 0$ the propagator  $e^{t\cL}$ is trace preserving and $\tr(\bX)=\tr\left(\coR\bX\right)$,  simple computations yield $\tr(\coM(\bX))=0$ and thus  $\tr(\coK(\bX))=\tr(\bX)$.

To prove complete-positivity, consider the extension of $\cL$, $\coK$, $\coR$ on the tensor product $\cH\otimes \widetilde{\cH}$ where $\widetilde{\cH}$ is any Hilbert space of finite dimension.  Let us prove that for $\bar\tau$ large enough,  such extension of $\coK$ is  non-negative, i. e., that for any $\ket{\Phi},\ket{\Psi}\in\cH\otimes \widetilde{\cH}$, we have
$
\bket{\Psi \Big| \coK\left(\ketbra{\Phi}{\Phi}\right)\Big| \Psi } \geq 0
.
$
Consider an Hilbert basis $\left(\ket{n}\right)_{1\leq n \leq d}$ of $\cH$ whose dimension is denoted by $d$. Take
$$
\ket{\Phi}= \sum_{n=1}^d \ket{n} \otimes \ket{\phi_n}, \quad \ket{\Psi}= \sum_{\nu=1}^d \ket{\nu} \otimes \ket{\psi_\nu}
$$
where, for each $n$ and $\nu$ in $\{1,\ldots, d\}$, $\ket{\phi_n},\ket{\psi_\nu} \in\widetilde{\cH}$.  Then standard computations give
$$
\bket{\Psi \Big| \coM\left(\ketbra{\Phi}{\Phi}\right)\Big| \Psi } =
 \sum_{n',\nu',n,\nu=1}^{d} z_{n',\nu'}^* ~M_{n',\nu',n,\nu} ~ z_{n,\nu}
$$
with $z_{n,\nu}= \braket{\phi_n}{\psi_\nu}$ and
$$
M_{n',\nu',n,\nu} = \int_0^{+\infty}\bket{\nu' \Big| e^{t\cL}(\ketbra{n'}{n}) - \coR(\ketbra{n'}{n})\Big| \nu}~dt
.
$$
Similarly
$$
\bket{\Psi \Big| \coR\left(\ketbra{\Phi}{\Phi}\right)\Big| \Psi } = \sum_{n',\nu',n,\nu=1}^{d} z_{n',\nu'}^* ~R_{n',\nu',n,\nu} ~ z_{n,\nu}
$$
with  $R_{n',\nu',n,\nu} = \bket{\nu' \Big|\coR(\ketbra{n'}{n})\Big| \nu}$.
This means that, defining the $d^2$-dimensional vector $z=(z_{n,\nu})_{n,\nu\in\{1,\ldots, d\}}$  and the $d^2\times d^2$ Hermitian matrices  $M=\Big(M_{n',\nu',~n,\nu}\Big)$  and $R=\Big(R_{n',\nu',~n,\nu}\Big)$ ,  we have  the following quadratic forms
$$
\bket{\Psi \Big| \coM\left(\ketbra{\Phi}{\Phi}\right)\Big| \Psi }= z^\dag M z, \quad \bket{\Psi \Big| \coR\left(\ketbra{\Phi}{\Phi}\right)\Big| \Psi }= z^\dag R z
$$
where $z$ depends on $\ket\Phi$ and $\ket\Psi$,   where  $M$ and $R$ depend only on $\coM$ and $\coR$.
We have thus  to prove that there exists $\bar\tau >0$ such that $M+\bar\tau R \geq 0$.  Since $\coR$ is a complete-positive map,  the $d^2\times d^2$ Hermitian matrix $R $ is non-negative.

Take $z$, such that $R z=0$. Take $T>0$. We have
\begin{align*}
 \bket{\Psi \Big| \int_0^{T}  e^{t\cL}\Big(\ketbra{\Phi}{\Phi} - \coR(\ketbra{\Phi}{\Phi} )\Big)~dt \Big| \Psi }
 &=\bket{\Psi \Big| \int_0^{T}  e^{t\cL}\Big(\ketbra{\Phi}{\Phi}\Big)~dt \Big| \Psi } \\
 &= \int_0^{T} \sum_{n,\nu,n',\nu'} z_{n',\nu'}^* \bket{\nu' \Big| e^{t\cL}(\ketbra{n'}{n}) \Big| \nu} z_{n,\nu}~dt
\end{align*}
since  $e^{t\cL}\Big(\coR(\ketbra{\Phi}{\Phi} ) \Big)=\coR(\ketbra{\Phi}{\Phi} ) $ and  \ $\bket{\Psi \Big| \coR(\ketbra{\Phi}{\Phi} ) \Big| \Psi }= z^\dag R z=0$.
Since for each $t\geq 0$, $e^{t\cL}$ is completely positive, then there exists a Kraus decomposition
$$
e^{t\cL}(\bX)= \sum_{\mu} \bW_{\mu,t} \bX \bW_{\mu,t}^\dag
$$
with operators $\bW_{\mu,t}$ on $\cH$ such that $\sum_{\mu} \bW_{\mu,t}^\dag \bW_{\mu,t} = \bI$.
We have
$$
 \bket{\nu' \Big| e^{t\cL}(\ketbra{n'}{n}) \Big| \nu}= \sum_\mu  \bket{\nu' \Big| \bW_{\mu,t}\Big| n'}  \bket{n \Big| \bW^\dag_{\mu,t}\Big| \nu}
 .
$$
Consequently
$$
\sum_{n,\nu,n',\nu'} z_{n',\nu'}^* \bket{\nu' \Big| e^{t\cL}(\ketbra{n'}{n}) \Big| \nu} z_{n,\nu}
= \sum_{\mu} \left| \sum_{n,\nu=1}^{d} \bket{ n\Big| W_{\mu,t}^\dag \Big | \nu} z_{n,\nu}\right|^2
.
$$
Since $z^\dag M z$ is the limit when $T$ tends to $+\infty$ of
$$
\bket{\Psi \Big| \int_0^{T}  e^{t\cL}\Big(\ketbra{\Phi}{\Phi}\Big)~dt \Big| \Psi },
$$
we have for any $T>0$
$$
z^\dag M z \geq \sum_{\mu} \int_{0}^T\left| \sum_{n,\nu=1}^{d} \bket{ n\Big| W_{\mu,t}^\dag \Big | \nu} z_{n,\nu}\right|^2~dt
\geq 0
$$
that is, M is non-negative definite. In particular, if we assume that  $z^\dag M z=0$. The above  inequality implies that for any $t >0$
$$
\sum_{n,\nu=1}^{d} \bket{ n\Big| W_{\mu,t}^\dag \Big | \nu} z_{n,\nu}=0
.
$$
Recall that, by assumption,
$
\sum_{n,\nu}  \bket{\nu' \Big| \coR(\ketbra{n'}{n})\Big| \nu}~ z_{n,\nu} =0$ for any $n',\nu'\in\{1,\ldots, d\}$. Consequently
\begin{eqnarray*}
  \sum_{n,\nu}  \left(\int_0^{T}\bket{\nu' \Big| e^{t\cL}(\ketbra{n'}{n}) - \coR(\ketbra{n'}{n})\Big| \nu}~dt \right) z_{n,\nu}
\\= \sum_{n,\nu}  \left(\int_0^{T}\bket{\nu' \Big| e^{t\cL}(\ketbra{n'}{n}) \Big| \nu}~dt \right) z_{n,\nu}
\\= \int_0^{T} \sum_\mu  \bket{\nu' \Big| \bW_{\mu,t} \Big| n'} \left(  \sum_{n,\nu} \bket{ n\Big| W_{\mu,t}^\dag \Big | \nu} z_{n,\nu}\right) =0
.
\end{eqnarray*}
Thus for any $z$ such that $R z=0$ and $z^\dag M z=0$, we have necessarily $Mz=0$ by taking the limit for $T$ tending to $+\infty$.

To summarize we have shown that
\begin{enumerate}
  \item $R \geq 0$ ;
  \item if  $z^\dag R z= 0$ then $z^\dag M  z \geq 0$;
  \item  if  $z^\dag R z=z^\dag M  z = 0$ then $M z=0$.
\end{enumerate}
According to Lemma~\ref{lem:Slem} proved below, there exists $\bar\tau >0$ such that $M + \bar\tau R \geq 0$.

\end{proof}

\begin{lemma} \label{lem:Slem}
Consider two Hermitian matrices of same dimension $R$ and $M$ such that $R$ is non negative,  such that $z^\dag R z=0$ implies that $z^\dag M z \geq 0$, and such $z^\dag R z=z^\dag M z = 0$ implies $Mz =0$. Then for $\tau\geq 0$ large enough, $M+ \tau R \geq 0$.
\end{lemma}
\begin{proof}
  Up to a unitary transformation, we have  the block decomposition associated with $\ker R$ and $\ker R^\perp$:
  $$
   M= \left(\begin{array}{cc}
       A & C^\dag \\
       C & B \\
     \end{array}\right), \quad
     R=  \left(\begin{array}{cc}
       0 & 0 \\
       0 & D \\
     \end{array}\right),
   \quad
   z=  \left(\begin{array}{c}
       x \\
       y \\
     \end{array}\right),
  $$
  with  $A$, $B$ and $D$ Hermitian matrices with $D>0$.
  For any $z$ such that  $z^\dag R z=0$ we have $z^\dag M z \geq 0$, this means that $A\geq 0$. Up to some unitary transformation on $A$ only, we can always assume the following  sub-block decomposition for $A$, $C$   and $x$
  $$
   A=  \left(\begin{array}{cc}
       0 & 0 \\
       0 & \bar A \\
     \end{array}\right),
     \quad
      C=  \left(\begin{array}{cc}
       \tilde C & \bar C \\
     \end{array}\right),
     \quad
     x=  \left(\begin{array}{c}
       \tilde x \\
       \bar x \\
     \end{array}\right),
  $$
  with $\bar A >0$.  According to these block decompositions, $z^\dag R z=z^\dag M z = 0$ means that $y=0$ and $\bar x=0$ with $\tilde x$ arbitrary.  But $Mz=0$ means  that
  $$
  \left(\begin{array}{cc}
       \tilde C & \bar C \\
     \end{array}\right)  \left(\begin{array}{c}
       \tilde x \\
       0 \\
     \end{array}\right)= \tilde C  \tilde x =0
     $$
     for all $\tilde x$. Thus $\tilde C=0$. To summarize, up to a unitary transformation, we have the following decomposition for $M$ and $R$
     $$
     M= \left(\begin{array}{ccc}
      0 & 0 & 0 \\
       0 & \bar A & \bar C^\dag \\
       0& \bar C & B \\
     \end{array}\right), \quad
     R=  \left(\begin{array}{ccc}
       0 & 0 &0 \\
       0 & 0 &0 \\
       0 & 0 &D \\
     \end{array}\right)
     $$
     with $\bar A>0$ and $D >0$. For $\tau$ large enough $ \tau \bar A \geq \bar C ^\dag D^{-1} \bar C$ and thus   $M+\tau R \geq 0$ (see, e.g, \cite[Theorem 1.3.3, page 14]{BhatiaBook2007bis}).
\end{proof}



\ra{We end this section by proving a factorization property with $\orho_A$ for the result of $\cL_A^{-1}$ used after \eqref{eq:InvLA} and then in Section \ref{sec:HamSyst}.} This property is a direct consequence of the following two short lemmas, first introduced in \cite{AzouitIFAC16}. For the sake of self-completeness we also present their proof. Then, we use these lemmas to derive Corollary \ref{cor:inclu_noyau} in the particular case of composite systems with a cascaded structure, as presented in Section \ref{sec:CascSyst}.

\begin{lemma} \label{lem:trivial}
Let $\cL_A(\orho_A) = 0$, with $\cL_A$ of the form \eqref{eq:DefLind_A} and $\orho_A$ a density operator. Then for any $\ket{\nu} \in \ker(\orho_A)$ we have $\sqrt{\orho_A} \bL_{A,\mu}^\dagger \ket{\nu} = 0$, for all $\mu$.
\end{lemma}
\begin{proof}

For $\ket{\nu} \in \ker(\orho_A)$ we have $\bra{\nu} \cL_A(\orho_A) \ket{\nu} = \sum_k \bra{\nu} \bL_{A,\mu} \orho_A \bL_{A,\mu}^\dagger \ket{\nu}$. Since $\cL_A(\orho_A) = 0$ each term of this positive sum must vanish.
\end{proof}
\smallskip

\begin{lemma} \label{lem:incluNoyau}
Denote by $\rho=\orho_A$ the unique density operator solution of  $\cL_A(\rho) = 0$. For a traceless operator $\bY$ such that $\ker(\orho_A) \subseteq \ker(\bY)$, the traceless solution to $\bX = \cL_A^{-1}(\bY)$ also satisfies $\ker(\orho_A) \subseteq \ker(\bX)$.
\end{lemma}
\begin{proof}
Note that the operators have such kernels if and only if they can be written
$\bX = \tilde{\bX} \orho_A$, $\bY = \tilde{\bY} \orho_A$. Since $\cL_A$ is a bijection on the space of traceless operators, the property is equivalent to show that $\bY \ket{\nu} = \cL_A(\tilde{\bX} \orho_A) \ket{\nu} = 0$ for all $\ket{\nu} \in \ker(\orho_A)$. By using $\orho_A \ket{\nu} = 0$ and Lemma \ref{lem:trivial}, we directly get
$$
\cL_A(\tilde{\bX} \orho_A) \ket{\nu} =
\bX \left(i \orho_A \bH_A - \tfrac{1}{2} \orho_A \sum_\mu \bL_{A_\mu}^\dagger \bL_{A_\mu}\right) \ket{\nu} \, .
$$
Subtracting $0 = \cL_A(\orho_A)$ inside the bracket, applying $\orho_A \ket{\nu} = 0$ and Lemma \ref{lem:trivial} once again, we do get 0.
\end{proof}

\begin{corollary} \label{cor:inclu_noyau}
Let the density operator $\orho_A$ be solution of $\cL_A(\orho_A) = \widetilde{\cL}_A(\orho_A) + \cD_{\ba}(\orho_A) = 0$ with $\widetilde{\cL}_A$ of the form \eqref{eq:DefLind_A} . For the traceless operator $\bY = \orho_A a^\dag - \tr(\orho_A a^\dag) \orho_A$, the traceless solution to $\bX = \cL_A^{-1}(\bY)$  satisfies $\ker(\orho_A) \subseteq \ker(\bX)$.
\end{corollary}

\begin{proof}
From Lemma \ref{lem:trivial} we get that for any $\ket{\nu} \in \ker(\orho_A)$ we have $\sqrt{\orho_A} \ba^\dagger \ket{\nu} = 0$. Therefore $\ker(\orho_A) \subseteq \ker(\bY)$, then apply Lemma \ref{lem:incluNoyau}.
\end{proof}


\section{Lindblad form of super-operators}


\subsection{$\cL_1$: partial trace of Lindblad super-operators}

On the Hilbert space $\cH_A\otimes \cH_B$ consider the Lindblad super-operator
$$
  \cL(\rho)= - i [\bH,\rho] +\sum_\mu \bL_{\mu}\rho \bL_{\mu}^\dag - \tfrac{1}{2} \Big( \bL_\mu^\dag \bL_\mu \rho + \rho  \bL_\mu^\dag \bL_\mu \Big)
  $$
with  $\bH $  an Hermitian operator on $\cH_A\otimes \cH_B$ and with  $\bL_{\mu}$ operators on  $\cH_A\otimes \cH_B$.
\begin{lemma} \label{lem:TrALint}
For any density operator $\rho_A$ on $\cH_A$ the super-operator given by
$$
\rho_B \mapsto \tr_{A}\Big( \cL(\rho_A\otimes \rho_B)\Big)
$$
remains a Lindbaldian  super-operator  on $\cH_B$.
\end{lemma}
\begin{proof}
  By linearity, it is enough to consider the following two cases:
  $$
  \cL(\rho)=  - i [\bH,\rho]\quad  \text{and} \quad \cL(\rho)= \bL\rho \bL^\dag - \tfrac{1}{2} \Big( \bL^\dag \bL \rho + \rho  \bL^\dag \bL \Big).
  $$
With $\bH= \sum_k \bA_k\otimes \bB_k$, we have
$
\tr_A\Big( [ \bH , \rho_A\otimes \rho_B]\Big) = \left[\bH_B, \rho_B\right]
$
where  $\bH_B= \sum_k\tr(\bA_k \rho_A) \bB_k$ is  an  Hermitian operator on $\cH_B$.

With $\bL= \sum_k \bA_k\otimes \bB_k$, we have
\begin{multline*}
  \tr_{A}\left( \bL \rho_A\otimes \rho_B \bL^\dag - \tfrac{1}{2} \Big( \bL^\dag \bL \rho_A\otimes \rho_B + \rho_A\otimes \rho_B  \bL^\dag \bL \Big)\right)
  \\
  =  \sum_{k,k'} \tr(\bA_k \rho_A \bA_{k'}^\dag)\left( \bB_k \rho_B \bB_{k'}^\dag - \tfrac{1}{2} \Big(\bB_{k'}^\dag \bB_k \rho_B + \rho_B \bB_{k'}^\dag \bB_k \Big)\right)
  .
\end{multline*}
Since $\rho_A \geq 0$, the square matrix $M=\Big( \tr(\bA_k \rho_A \bA_{k'}^\dag)\Big)$ is Hermitian an non-negative (use the spectral decomposition of $\rho_A$ and the fact that, for any $\ket{\psi}\in\cH_A$, the matrix  $\Big( \tr(\bA_k \ket\psi\bra\psi  \bA_{k'}^\dag)\Big)$ is the  Gram matrix  associated with the vectors $\bA_k\ket\psi$ of $\cH_A$). Take the   Cholesky decomposition  $M= N N^\dag $.  We have
 \begin{multline*}
   \sum_{k,k'} M_{k,k'} \left( \bB_k \rho_B \bB_{k'}^\dag - \tfrac{1}{2} \Big(\bB_{k'}^\dag \bB_k \rho_B + \rho_B \bB_{k'}^\dag \bB_k \Big)\right)
    \\
    = \sum_{k,k',k''} N_{k,k"} N^*_{k',k''} \left( \bB_k \rho_B \bB_{k'}^\dag - \tfrac{1}{2} \Big(\bB_{k'}^\dag \bB_k \rho_B + \rho_B \bB_{k'}^\dag \bB_k \Big)\right)
    \\
   = \sum_{k''} \bX_{k''} \rho_B \bX_{k''}^\dag - \tfrac{1}{2} \Big(\bX_{k''}^\dag \bX_{k''} \rho_B + \rho_B \bX_{k''}^\dag \bX_{k''} \Big)
\end{multline*}
with $\bX_{k''}= \sum_{k} N_{k,k''} \bB_k$.
\end{proof}


\subsection{$\cL_2$: positivity of quadratic form on $\{\bB_k,\bB_k^\dagger\}$}

The proof uses the following property.

\begin{lemma} \label{Lem:Xipos-Aux1}:
For any operators $\bX$ and $\bY$ on $\cH_A$, it holds that
$$
\cL_A(\bX \orho_A) \bY^\dagger + \bX \cL_A(\orho_A \bY^\dagger) =
\cL_A(\bX \orho_A \bY^\dagger) - \sum_\mu [\bL_{A,\mu},\bX] \orho_A [\bL_{A,\mu},\bY]^\dagger \, .
$$
\end{lemma}
\begin{proof}
To check this, just write down the expressions, cancel some terms and use $\cL_A(\orho_A) = 0$ for the remaining ones.
\end{proof}

\begin{lemma}\label{Lem:Xipos}
The matrix $X$ in \eqref{eq:L2bigfirst} is always positive semidefinite. In particular, for $\Hint=\bA \otimes \bB$ with $\bA$ and $\bB$ Hermitian, the coefficient $\tr(\bF \orho_A \bA + \bA \orho_A \bF^\dagger)$ is always nonnegative.
\end{lemma}
\begin{proof}
The proof for the particular case is simple and introduces the main idea. Apply Lemma \ref{Lem:Xipos-Aux1} with $\bX = \bY = \bF$; use that $-\cL_A(\bF \orho_A) = \bA \orho_A - \tr(\bA \orho_A) \orho_A$; take the trace over $\cH_A$ and use that $\tr(\bF \orho_A) = 0 = \tr(\cL_A)$. This yields:
$$\tr(\bF \orho_A \bA + \bA \orho_A \bF^\dagger) = \sum_\mu \; \tr([\bL_{A,\mu},\bF] \orho_A [\bL_{A,\mu},\bF]^\dagger)$$
where the right hand side is obviously nonnegative.

Applying the same idea with $(\bX,\bY) = (\bF_j,\bF_k)$, shows that the components of $X$ can be re-expressed as
\begin{eqnarray*}
X_{j,k} & = & \sum_{\mu} \tr\left([\bL_{A,\mu},\bF_k]\,\orho_A\,[\bL_{A,\mu},\bF_j]^\dagger\right) =: \sum_{\mu} \; X_{j,k}^{(\mu,\orho_A)} .
\end{eqnarray*}
For each $\mu$ and replacing $\orho_A$ by any pure state $\ket{\opsi}\bra{\opsi}$, we would thus have
$$X_{j,k}^{(\mu,\ket{\opsi}\bra{\opsi})} = \bra{v_j^{(\mu,\opsi)}} \ket{v_k^{(\mu,\opsi)}} \quad \text{with } \ket{v_k^{(\mu,\opsi)}} = [\bL_{A,\mu},\bF_k] \ket{\opsi} \, .$$
The corresponding matrix $X^{(\mu,\ket{\opsi}\bra{\opsi})}$ is a Gram matrix, which is always positive semi-definite, see e.g.~\cite[exercise 1.1.1, page 3]{BhatiaBook2007bis}. The same then obviously holds for $X$, which is obtained by taking the sum over $\mu$ and a convex combination over different $\ket{\opsi}\bra{\opsi}$.
\end{proof}


\section{Computations for the cascaded example} \label{apx:CalcSqueeze}

We detail here the computations of the coefficients involved in the second order reduced dynamics of a cavity with a squeezed drive coupled to an unspecified subsystem presented in Section \ref{subsec:Squeezed}.
The fast dynamics is given by
\begin{align*}
\frac{d}{dt} \rho = g[\ba^2-{\ba^\dag}^2, \rho] + \kappa \cD_{\ba}(\rho)
\end{align*}
with $\kappa > 4g$ to guarantee the convergence towards a steady state. The evolution of the operator $\ba$ in the Heisenberg picture reads:
\begin{align} \label{eq:ExSqueeze_a}
\frac{d e^{t \cL_A^*}(\ba)}{dt}=-g[\ba^2-{\ba^\dag}^2, e^{t \cL_A^*}(\ba)] + \kappa \cD_{\ba}^*(e^{t \cL_A^*}(\ba))
\end{align}
where $\cD_{\ba}^*$ is the dual of $\cD_{\ba}$. We search $e^{t \cL_A^*}(\ba)$ in the particular form $e^{t \cL_A^*}(\ba) = f_a(t) \ba + h_a(t) \ba^\dag$ for some scalar functions $f_a(t)$ and $h_a(t)$ verifying $f_a(0)=1$ and $h_a(0) = 0$. Injecting this expression into \eqref{eq:ExSqueeze_a} and solving the differential equations for $f_a(t)$ and $h_a(t)$ leads to $f_a(t) = \frac{e^{-t \frac{\kappa-4g}{2}}+e^{-t \frac{\kappa+4g}{2}}}{2}$, $h_a(t) = \frac{e^{-t \frac{\kappa+4g}{2}}-e^{-t \frac{\kappa-4g}{2}}}{2}$ and thus for any initial state $\rho_0$, $\tr\bigl(e^{\infty \cL_A^*}(\ba) \rho_0\bigr)=0$. As the computation of the coefficients $\alpha$ and $\beta$ is similar, we detail here only the calculation of $\alpha=\kappa \tr \biggl( e^{\infty \cL_A^*}\bigl( \oa^\dag \int_0^\infty e^{t \cL_A^*}(\ba) dt \bigr)~ \rho_0 \biggr)$. Using the previous computation and by linearity of the integral we get:
\begin{align*}
\alpha = \kappa \tr \biggl( e^{\infty \cL_A^*}(\ba^\dag \ba) \rho_0 \biggr) \int_0^\infty f_a(t) dt + \kappa \tr \biggl( e^{\infty \cL_A^*}(\ba^\dag \ba^\dag) \rho_0 \biggr) \int_0^\infty g_a(t) dt
\end{align*}
Then pick $e^{t \cL_A^*}(\ba^\dag \ba)=f_N(t) \bN + h_N(t) \ba^2 + l_N(t) {\ba^\dag}^2 + m_N(t)\bI $,  (similarly for $e^{t \cL_A^*}(\ba^\dag \ba^\dag)$) and standard computations lead to $\alpha = \frac{32\kappa^2 g^2}{((\kappa+4g)(\kappa-4g))^2}$.

\bibliographystyle{plain}

\end{document}